\documentclass[11pt]{article}
\usepackage{fullpage}

\usepackage{hyperref}

\usepackage{times}
\usepackage{helvet}
\usepackage{courier}
\usepackage{color}

\usepackage[ruled, boxed]{algorithm}
\usepackage{algorithmic}
\usepackage{amsthm,amsmath, amssymb}
\usepackage{enumerate}
\usepackage{comment}

\newcommand{\cM}{\mathcal{M}}
\newcommand{\cG}{\mathcal{G}}
\renewcommand{\varepsilon}{\epsilon}

\newtheorem{theorem}{Theorem}
\newtheorem{lemma}{Lemma}
\newtheorem{definition}{Definition}

\theoremstyle{definition}

\hypersetup{
     colorlinks   = true,
     urlcolor    = blue,
	 citecolor = blue
}

\title{An Algorithmic Framework for Strategic Fair Division\footnote{Br\^{a}nzei was supported
		by the Danish National Research Foundation
		and The National Science Foundation of China (under the grant 61361136003) for
		the Sino-Danish Center for the Theory of Interactive Computation, and the Center for
		Research in Foundations of Electronic Markets (CFEM), supported by the Danish
		Strategic Research Council. Br\^{a}nzei was also supported by ISF grant 1435/14 administered by the Israeli Academy of Sciences and Israel-USA Bi-national Science
		Foundation (BSF) grant 2014389, as well as the I-CORE Program of the Planning and Budgeting Committee and The Israel
		Science Foundation. Kurokawa and Procaccia were partially supported by the National Science Foundation under grants CCF-1215883, CCF-1525932, and IIS-1350598, and by a Sloan Research Fellowship.
		Caragiannis was partially supported by Caratheodory research grant E.114 from the University of Patras.
	}}

\author{
	\textbf{Simina Br\^anzei}\footnote{E-mail: \texttt{simina.branzei@gmail.com}}\\
	Hebrew University of Jerusalem\\
	\newline
	\and
	\textbf{Ioannis Caragiannis}\footnote{E-mail: \texttt{caragian@ceid.upatras.gr}}\\
	University of Patras\\
	\and
	\textbf{David Kurokawa}\footnote{E-mail: \texttt{dkurokaw@cs.cmu.edu}}\\
	Carnegie Mellon University\\
	\and
	\textbf{Ariel D. Procaccia}\footnote{E-mail: \texttt{arielpro@cs.cmu.edu}}\\
	Carnegie Mellon University
}
\date{}

\begin{document}
	\maketitle

\begin{abstract}

We study the paradigmatic fair division problem of allocating a divisible good among agents with heterogeneous preferences, commonly known as \emph{cake cutting}. Classical cake cutting protocols are susceptible to manipulation. Do their strategic outcomes still guarantee fairness? 

To address this question we adopt a novel algorithmic approach, by designing a concrete computational framework for fair division --- the class of \emph{generalized cut and choose (GCC) protocols} --- and reasoning about the game-theoretic properties of algorithms that operate in this model. The class of GCC protocols includes the most important discrete cake cutting protocols, and turns out to be compatible with the study of fair division among strategic agents. In particular, GCC protocols 
are guaranteed to have approximate subgame perfect Nash equilibria, or even exact equilibria if the protocol's tie-breaking rule is flexible. We further observe that the (approximate) equilibria of proportional GCC protocols --- which guarantee each of the $n$ agents a $1/n$-fraction of the cake --- must be (approximately) proportional. Finally, we design a protocol in this framework with the property that its Nash equilibrium allocations coincide with the set of (contiguous) envy-free allocations.
\end{abstract}

\newpage

\section{Introduction}

How should one allocate resources among economic agents with heterogeneous preferences, in a way that is fair to everyone involved? 
The rigorous study of this question dates back to the 1940's; it has led to an extensive body of literature in economics, mathematics, and political science, with several books written on the topic~\cite{Moul03,RW98,BT96,Young94}. More recently, fair division has become an important problem in computer science~\cite{Pro09,CLP11,CLPP11,DFHK+11,BFML+12,BCHT+12,Cavallo12,ADH13,KLP13,BPZ13,CLPP13,Pro13,PW14,BBKP+14,SHA15,SHA15b,ABFF15,AMNS15,CG15,AM16}; specific focus areas include the allocation of computational resources~~\cite{GZHK+11,GN12,KPS13,CGG13,PPS15}, and the implementation and deployment of practical, provably fair solutions for real-world problems~\cite{OSB10,AAGW15,GP14,KPS15,GMPZ16,CKMP+16}.

Among problems in fair division, the so-called \emph{cake cutting} problem is perhaps the most paradigmatic. It was formalized by Steinhaus~\cite{Stein48} during World War II and studied in a rich body of literature over the years. The \emph{cake} is a metaphor for a heterogeneous divisible resource, such as land, time, memory in shared computing systems, clean water, greenhouse gass emissions, or fossil fuels. 

Going back to the word ``fair'', two formal notions of fairness have emerged as the most appealing and well-studied in the context of cake cutting: \emph{proportionality}, in which each of the $n$ agents receives at least a $1/n$-fraction of the entire cake according to its valuation; and \emph{envy-freeness}, which stipulates that no agent would wish to swap its own piece with that of another agent. At the heart of the cake cutting endeavor is the design of cake cutting \emph{protocols}, which specify an interaction between agents --- typically via iterative steps of manipulating the cake --- such that the final allocation is guaranteed to be proportional or envy-free.

The simplest cake cutting protocol is known as \emph{cut and choose}, and is designed for two agents. The first agent cuts the cake in two pieces that it values equally; the second agent then chooses the piece that it prefers, leaving the first agent with the remaining piece. It is easy to see that this protocol yields a proportional and envy-free allocation (in fact these two notions coincide when there are only two agents and the entire cake is allocated). However, taking a game-theoretic point of view, it is immediately apparent that the agents can often do better by disobeying the protocol when they know each other's valuations. For example, in the cut and choose protocol, assume the first agent only desires a specific small piece of cake, whereas the second agent uniformly values the cake. The first agent can obtain its entire desired piece, instead of just half of it, by carving that piece out.

So how would strategic agents behave when faced with the cut and choose protocol? A standard way of answering this question employs the notion of \emph{Nash equilibrium}: each agent would use a strategy that is a best response to the other agent's strategy. 
To set up a Nash equilibrium, suppose that the first agent cuts two pieces that the second agent values equally; the second agent selects its more preferred piece, and the one less preferred by the first agent in case of a tie. Clearly, the second agent cannot gain from deviating, as it is selecting a piece that is at least as preferred as the other. As for the first agent, if it makes its preferred piece even bigger, the second agent would choose that piece, making the first agent worse off. 
Interestingly enough, in this equilibrium the tables are turned; now it is the second agent who is getting exactly half of its value for the whole cake, while the first agent generally gets more. Crucially, the equilibrium outcome is also proportional and envy-free. In other words, even though the agents are strategizing rather than following the protocol, the outcome in equilibrium has the same fairness properties as the ``honest'' outcome! 

With this motivating example in mind, we would like to make general statements regarding the equilibria of cake cutting protocols. We wish to identify a general family of cake cutting protocols --- which captures the classic cake cutting protocols --- so that each protocol in the family is guaranteed to possess (approximate) equilibria. Moreover, we wish to argue that these equilibrium outcomes are fair. Ultimately, our goal is to be able to reason about the fairness of cake divisions that are obtained as outcomes when agents are presented with a standard cake cutting protocol and behave strategically. 

\subsection{Model and Results}

To set the stage for a result that encompasses classic cake cutting protocols, we introduce (in Section~\ref{sec:model}) the class of \emph{generalized cut and choose (GCC)} protocols. A GCC protocol is represented by a tree, where each node is associated with the action of an agent. The tree has two types of nodes: a \emph{cut node}, which instructs the agent to make a cut between two existing cuts; and a \emph{choose} node, which offers the agent a choice between a collection of pieces that are induced by existing cuts. Moreover, we assume that the progression from a node to one of its children depends only on the relative positions of the cuts (in a sense to be explained formally below). We argue that classic protocols --- such as Dubins-Spanier~\cite{DS61}, Selfridge-Conway (see \cite{RW98}), Even-Paz~\cite{EP84}, as well as the original cut and choose protocol --- are all GCC protocols.

We view the definition of the class of GCC protocols as one of our main conceptual contributions, since cake cutting 
protocols have not enjoyed a computational model until this work. 

In Section~\ref{sec:exist}, we observe that GCC protocols may not have exact Nash equilibria (NE), then explore ways of circumventing this issue, which give rise to our first two main results. 
\begin{itemize}
	\item We prove that every GCC protocol has at least one $\varepsilon$-NE for every $\varepsilon>0$, in which agents cannot gain more than $\varepsilon$ by deviating, and $\varepsilon$ can be chosen to be arbitrarily small. In fact, we establish this result for a stronger equilibrium notion, (approximate) \emph{subgame perfect Nash equilibrium (SPNE)}, which is, intuitively, a strategy profile where the strategies are in NE even if the game starts from an arbitrary point. 
	
	\item We slightly augment the class of GCC protocols by giving them the ability to make \emph{informed tie-breaking} decisions that depend on the entire history of play, in cases where multiple cuts are made at the exact same point. While, for some valuation functions of the agents, a GCC protocol may not possess any exact SPNE, we prove that it is always possible to modify the protocol's tie-breaking scheme to obtain SPNE. 
\end{itemize}

In Section~\ref{sec:ef}, we observe that for any proportional protocol, the outcome in any $\varepsilon$-equilibrium must be an $\varepsilon$-proportional division. We conclude that under the classic cake cutting protocols listed above --- which are all proportional --- strategic behavior preserves the proportionality of the outcome, either approximately, or exactly under informed tie-breaking. 

One may wonder, though, whether an analogous result is true with respect to envy-freeness. We give a negative answer, by constructing an envy-inducing SPNE under the Selfridge-Conway protocol, a well-known envy-free protocol for three agents. 

However, our third main result is the construction of a GCC protocol in which every NE outcome is a contiguous envy-free allocation and vice versa, that is, the set of NE outcomes coincides with the set of contiguous envy-free allocations. This shows that the GCC framework is compatible with arguably the most important notion of fairness, namely envy-freeness. It remains open whether a similar result can be obtained for SPNE instead of NE. 


\subsection{Related Work}

The notion of GCC protocols is inspired by the Robertson-Webb~\cite{RW98} model of cake cutting --- a concrete query model that specifies how a cake cutting protocol may interact with the agents. This model underpins a significant body of work in theoretical computer science and artificial intelligence, which focuses on the complexity of achieving different fairness or efficiency notions in cake cutting~\cite{EP06b,EP06,WS07,DQS12,ADH13,Pro09,KLP13}. In Section~\ref{sec:model}, we describe the Roberston-Webb model in detail, and explain why it is inappropriate for reasoning about equilibria. 

In the context of the strategic aspects of cake cutting, Nicol\`o and Yu~\cite{NY08} were the first to suggest equilibrium analysis for cake cutting protocols. Focusing exclusively on the case of two agents, they design a specific cake cutting protocol whose unique SPNE outcome is envy-free. And while the original cut and choose protocol also provides this guarantee, it is not ``procedural envy free'' because the cutter would like to exchange roles with the chooser; the two-agent protocol of Nicol\'o and Yu aims to solve this difficulty. Br\^anzei and Miltersen~\cite{BM13} also investigate equilibria in cake cutting, but in contrast to our work they focus on one cake cutting protocol --- the Dubins-Spanier protocol --- and restrict the space of possible strategies to \emph{threshold strategies}. Under this assumption, they characterize NE
outcomes, and in particular they show that in NE the allocation is envy-free. Br\^anzei and Miltersen also prove the existence of $\varepsilon$-equilibria that are $\varepsilon$-envy-free; again, this result relies on their strong restriction of the strategy space, and applies to one specific protocol.

Several papers by computer scientists~\cite{CLPP13aaai,MT10,MN12} take a mechanism design approach to cake cutting; their goal is to design cake cutting protocols that are \emph{strategyproof}, in the sense that agents can never benefit from manipulating the protocol. This turns out to be an almost impossible task~\cite{Zhou91,BM15}; positive results are obtained by either making extremely strong assumptions (agents' valuations are highly structured), or by employing randomization and significantly weakening the desired properties. In contrast, our main results, given in Section~\ref{sec:exist}, deal with strategic outcomes under a large class of cake cutting protocols, and aim to capture well-known protocols; our result of Section~\ref{sec:ef} is a positive result that achieves fairness ``only'' in equilibrium, but without imposing any restrictions on the agents' valuations.


\section{The Model}
\label{sec:model}


The cake cutting literature typically represents the cake as the interval $[0,1]$. There is a set of agents $N=\{1,\ldots,n\}$, and each agent $i\in N$ is endowed with a \emph{valuation function} $V_i$ that assigns a value to every subinterval of $[0,1]$. These values are induced by a non-negative continuous \emph{value density function} $v_i$, so that for an interval $I$, $V_i(I)=\int_{x\in I} v_i(x)\ dx$. By definition, $V_i$ satisfies the first two properties below; the third is an assumption that is made w.l.o.g. 
\begin{enumerate}
	\item Additivity: For every two disjoint intervals $I_1$ and $I_2$, $V_i(I_1\cup  I_2) = V_i(I_1) + V_i(I_2)$.
	\item Divisibility: For every interval $I\subseteq [0,1]$ and $0\leq \lambda\leq 1$ there is a subinterval $I'\subseteq I$ such that $V_i(I') = \lambda V_i(I)$.
	\item Normalization: $V_i([0,1])=1$.
\end{enumerate}

Note that valuation functions are non-atomic, i.e., they assign zero value to points. This allows us to disregard the boundaries of intervals, and in particular we treat intervals that overlap at their boundary as disjoint. 
We sometimes explicitly assume that the value density functions are \emph{strictly positive}, i.e., $v_i(x)>0$ for all $x\in[0,1]$ and for all $i \in N$; this implies that $V_i([x,y])>0$ for all $x <y, x,y \in[0,1]$.

A \emph{piece of cake} is a finite union of disjoint intervals. We are interested in allocations of disjoint pieces of cake $X_1,\ldots,X_n$, where $X_i$ is the piece allocated to agent $i\in N$. A piece is \emph{contiguous} if it consists of a single interval.

We study two fairness notions. An allocation $X$ is \emph{proportional} if for all $i\in N$, $V_i(X_i)\geq 1/n$; and \emph{envy-free} if for all $i,j\in N$, $V_i(X_i)\geq V_i(X_j)$. Note that envy-freeness implies proportionality when the entire cake is allocated.

\subsection{Generalized Cut and Choose Protocols}
\label{sec:gcc}

The standard communication model in cake cutting was proposed by Robertson and Webb \cite{RW98}; it restricts the interaction between the protocol and agents to two types of queries:
\begin{itemize}
	\item \emph{Cut} query: $\text{\emph{Cut}}_{i}(x, \alpha)$ asks agent $i$ to return a point $y$ such that $V_{i}([x,y]) = \alpha$.
	\item \emph{Evaluate} query: $\text{\emph{Evaluate}}_{i}(x, y)$ asks agent $i$ to return a value $\alpha$ such that $V_{i}([x,y]) = \alpha$, where the points $x,y$ are either $0,1$, or have been generated as answers to previous \emph{Cut} queries.
	
\end{itemize}
Note that in the RW model,
a protocol could allocate pieces depending on whether a particular cut was made at a specific point (see Algorithm \ref{alg:nonGCC}). More generally, a protocol in the RW model has a property such as envy-freeness if, roughly speaking, it gathers enough information so that there \emph{exists} an allocation such that for \emph{any} valuation function consistent with the answers to the queries, the allocation is envy-free. Since the RW model does not specify how the allocation is computed, there need not exist a succinct representation of the allocation that arises as the outcome of a protocol, which makes it difficult to analyze the strategic properties of protocols in the RW model. 

%

For this reason, we define a generic class of protocols that are implementable with natural operations, which capture all bounded\footnote{In the sense that the number of operations is upper-bounded by a function that takes the number of agents $n$ as input.} and discrete cake cutting algorithms,
such as cut and choose, Dubins-Spanier, Even-Paz, Successive-Pairs, and Selfridge-Conway (see, e.g., \cite{Pro13}). At a high level, the standard protocols
are implemented using a sequence of natural instructions, each of which is either a
\emph{Cut} operation, in which some agent is asked to make a cut
in a specified region of the cake; or a \emph{Choose} operation, in which some agent is asked to take a piece from a set of already demarcated
pieces indicated by the protocol.
In addition, every node in the decision tree of the protocol is based exclusively on the execution history and absolute ordering of the cut points,
which can be verified with any of the following operators: $<, \leq, = , \geq, >$.

Formally, a \emph{generalized cut and choose (GCC)} protocol is implemented exclusively with the following instructions:
\begin{itemize}
	\item \emph{Cut}: The syntax is ``$i$ \emph{Cuts} in $S$'',
	where $S = \{[x_1, y_1], \ldots, [x_m, y_m]\}$ is a set of contiguous pieces (intervals), such that the endpoints of every piece $[x_j, y_j]$ are $0, 1$, or cuts made in previous steps of the protocol. Agent $i$ can make a cut at any point $z \in [x_j, y_j]$, for some $j \in \{1, \ldots, m\}$.
	
	\item \emph{Choose}: The syntax is ``$i$ \emph{Chooses} from $S$'',
	where $S = \{[x_1, y_1], \ldots, [x_m, y_m]\}$ is a set of contiguous pieces, such that the endpoints of every piece $[x_j, y_j] \in S$ are $0$, $1$, or cuts made in the previous steps of the protocol. Agent $i$ can choose any \emph{single} piece $[x_j, y_j]$ from $S$ to keep from that point on.
	
	\item \emph{If-Else Statements}: The conditions depend on the result of choose queries and the absolute order of all the cut points made in the previous steps.
\end{itemize}


A GCC protocol uniquely identifies every contiguous piece by the symbolic names of all the cut points contained in it.
For example, Algorithm 1 is a GCC protocol. Algorithm \ref{alg:nonGCC} is not a GCC protocol, because
it verifies that the point where agent $1$ made a cut is exactly $1/3$, whereas a GCC protocol can only
verify the ordering of the cut points relative to each other and the endpoints of the cake. Note that, unlike in the 
communication model of Robertson and Web \cite{RW98}, GCC protocols cannot obtain and use information about the valuations of the agents --- the allocation is only decided by the agents' \emph{Choose} operations.

\begin{algorithm}[t]
	\begin{algorithmic}
		\STATE agent $1$ \emph{Cuts} in $\{[0,1]\}$ // $@x$
		\STATE agent $1$ \emph{Cuts} in $\{[0,1]\}$ // $@y$
		\STATE agent $1$ \emph{Cuts} in $\{[0,1]\}$ // $@z$
		\IF {$\left(x < y < z \right)$}
		\STATE agent $1$ \emph{Chooses} from $\{ [x,y], [y, z]\}$
		\ENDIF
	\end{algorithmic}
	\caption{A GCC protocol. The notation ``// $@x$'' assigns the symbolic name $x$ to the cut point made by agent 1.}
	\label{alg:GCC_protocol}
\end{algorithm}

\begin{algorithm}[t]
	\begin{algorithmic}
		\STATE agent $1$ \emph{Cuts} in $\{[0,1]\}$ // $@x$
		\IF{$\left(x = \frac{1}{3}\right)$}
		\STATE agent $1$ \emph{Chooses} from $\{[0, x], [x,1]\}$
		\ENDIF
	\end{algorithmic}
	\caption{A non-GCC protocol.}
	\label{alg:nonGCC}
\end{algorithm}

As an illustrative example, we now discuss why the discrete variant of the Dubins-Spanier protocol\footnote{In fact, the discrete variant of Dubins-Spanier was invented much earlier by Banach and Knaster and is better known as the ``last diminisher'' procedure (see Steinhaus \cite{Stein48}).} belongs to the class of GCC protocols --- but first we must describe the original Dubins-Spanier protocol. Dubins-Spanier is a proportional (but not envy-free) protocol for $n$ agents, which operates in $n$ rounds. In round 0, each agent makes a mark $x_{i}^{1}$ such that the piece of cake to the left of the mark is worth $1/n$, i.e., $V_i([0,x_i^1])=1/n$. Let $i^*$ be the agent that made the leftmost mark; the protocol allocates the interval $[0,x_{i^*}^1]$ to agent $i^*$; the allocated interval and satisfied agent are removed. In round $t$, the same procedure is repeated with the remaining $n-t$ agents and the remaining cake. When there is only one agent left, it receives the remaining cake. To see why the protocol is proportional, first note that in round $t$ the remaining cake is worth at least $1-t/n$ to each remaining agent, due to the additivity of the valuation functions and the fact that the pieces allocated in previous rounds are worth at most $1/n$ to these agents. The agent that made the leftmost mark receives a piece that it values at $1/n$. In round $n-1$, the last agent is left with a piece of cake worth at least $1-(n-1)/n=1/n$. 

The protocol admits a GCC implementation
as follows. For the first round, each agent $i$ is required to make a cut in $\{[0,1]\}$, at some point denoted by $x_{i}^{1}$. The agent $i^*$ with the leftmost cut $x_{i^*}^{1}$ can be determined using \emph{If-Else} statements whose conditions only
depend on the ordering of the cut points $x_{1}^{1}, \ldots, x_{n}^{1}$.
Then, agent $i^*$ is asked to choose ``any'' piece in the singleton set $\{[0,x_{i^*}^{1}]\}$.
The subsequent rounds are similar: at the end of every round the agent that was allocated a piece is removed, and the protocol iterates on the remaining agents and remaining cake. Note that agents are not constrained to follow the protocol, i.e., they can make their marks (in response to cut instructions) wherever they want; nevertheless, an agent can guarantee a piece of value at least $1/n$ by following the Dubins-Spanier protocol, regardless of what other agents do. 

%

While GCC protocols are quite general, a few well-known cake cutting protocols are beyond their reach. For example, the Brams-Taylor~\cite{BT95} protocol is an envy-free protocol for $n$ agents, and although its individual operations are captured by the GCC formalism, the number of operations is not bounded as a function of $n$ (i.e., it may depend on the valuation functions themselves). Its representation as a GCC protocol would therefore be infinitely long. 
In addition, some cake cutting protocols use \emph{moving knives} (see, e.g., \cite{BTZ97}); for example, they can keep track of how an agent's value for a piece changes as the piece smoothly grows larger. These protocols are not discrete, and, in fact, cannot be implemented even in the Robertson-Webb model. 

We also note that the GCC model is \emph{incomparable} to the RW model. Indeed, given a protocol in the RW model, it may not be possible to implement it as a GCC protocol because the RW model does not indicate a specific allocation, as discussed above. Conversely, cut queries in the GCC model cannot in general be translated into cut queries in the RW model, as in the latter model cuts are associated with a specific value. 

\subsection{The Game}

We study GCC protocols when the agents behave strategically. Specifically, we consider a GCC protocol, coupled with the valuation functions of the agents, 
as an \emph{extensive-form game of perfect information} (see, e.g., \cite{SL08}). In such a game, agents execute the \emph{Cut} and \emph{Choose} instructions strategically. 
Each agent is fully aware of the valuation functions of the other agents and aims to optimize its overall utility for the chosen pieces, given the strategies of other agents. 

While the perfect information model may seem restrictive, the same assumption is also made in previous work on equilibria in cake cutting~\cite{NY08,BM13}. More importantly, it underpins foundational papers in a variety of areas of microeconomic theory, such as the seminal analysis of the Generalized Second Price (GSP) auction by Edelman et al.~\cite{EOS07}. A common justification for the complete information setting, which is becoming increasingly compelling as access to big data gets pervasive, is that agents can obtain significant amounts of information about each other from historical data. 

In more detail, the game can be represented by a tree (called a \emph{game tree}) with Cut and Choose nodes:
\begin{itemize}
	\item In a \emph{Cut} node defined by ``\emph{$i$ cuts in $S$}'', where $S = \{[x_1, y_1], \ldots, [x_m, y_m]\}$, the strategy space of agent $i$ is the set $S$ of points where $i$ can make a cut at this step.
	\item  In a \emph{Choose} node defined by ``\emph{$i$ chooses from $S$}'', where $S = \{[x_1, y_1], \ldots, [x_m, y_m]\}$, the strategy space is the set $\{1, \ldots, m\}$, i.e., the indices of the pieces that can be chosen by the agent from the set $S$.
\end{itemize}

The strategy of an agent defines an action for \emph{each} node of the game tree where it executes a \emph{Cut} or a \emph{Choose} operation. If an agent deviates, the game can follow a completely different branch of the tree, but the outcome will still be well-defined. 

The strategies of the agents are in \emph{Nash equilibrium (NE)} if no agent can improve its utility by unilaterally deviating from its current strategy, i.e., by cutting at a different set of points and/or by choosing different
pieces. A \emph{subgame perfect Nash equilibrium (SPNE)} is a stronger equilibrium notion, which means that the strategies are in NE in every subtree of the game tree. In other words, even if the game started from an arbitrary node of the game tree, the strategies would still be in NE. An \emph{$\varepsilon$-NE} (resp., $\varepsilon$-SPNE) is a relaxed solution concept where an agent cannot gain more than $\varepsilon$ by deviating (resp., by deviating in any subtree).


\section{Existence of Equilibria}
\label{sec:exist}

It is well-known that finite extensive-form games of perfect information can be solved using \emph{backward induction}: starting from the leaves and progressing towards the root, at each node the relevant agent chooses an action that maximizes its utility, 
given the actions that were computed for the node's children. The induced strategies form an SPNE. Unfortunately, although we consider finite GCC protocols, we also need to deal with \emph{Cut} nodes where the action space is infinite, hence na\"ive backward induction does not apply. 

In fact, it turns out that not every GCC protocol admits an exact NE --- not to mention SPNE. For example, consider Algorithm 1, and assume that the value density function of agent 1 is strictly positive.
Assume there exists a NE where agent $1$ cuts at  $x^{*}, y^{*}, z^{*}$, respectively,
and chooses the piece $[x^{*}, y^{*}]$.
If $x^{*} > 0$, then the agent can improve its utility by making the first cut at $x' = 0$ and choosing the piece $[x', y^{*}]$,
since $V_1([x', y^{*}]) > V_{1}([x^{*}, y^{*}])$. Thus, $x^{*} = 0$.
Moreover, it cannot be the case that $y^{*} = 1$, since the agent only receives an allocation
if $y^{*} < z^{*} \leq 1$. Thus, $y^{*} < 1$. Then, by making the second cut at any $y' \in (y^{*}, z^{*})$, agent $1$ can obtain
the value
$V_{1}([0, y']) > V_{1}([0, y^{*}])$.
It follows that there is no exact NE where the agent chooses the first piece. Similarly, it can be shown that there is no exact NE where the
agent chooses the second piece, $[y^{*}, z^{*}]$. This illustrates why backward induction does not apply: the maximum value at some \emph{Cut} nodes may not be well defined.

\subsection{Approximate SPNE}
\label{subsec:approx}

One possible way to circumvent the foregoing example is by saying that agent 1 should be happy to make the cut $y$ very close to $z$. For instance, if the agent's value is uniformly distributed over the case, cutting at $x=0, y=1-\epsilon, z=1$ would allow the agent to choose the piece $[x,y]$ with value $1-\epsilon$; and this is true for any $\epsilon$. 

More generally, we have the following theorem. 


\begin{theorem} \label{thm:existence}
	For any $n$-agent GCC protocol $\mathcal{P}$ with a bounded number of steps, any $n$ valuation functions $V_1,\ldots,V_n$, and any $\epsilon>0$, the game induced by $\mathcal{P}$ and $V_1,\ldots,V_n$ has an $\varepsilon$-SPNE.
\end{theorem}

The proof of Theorem~\ref{thm:existence} is relegated to \ref{app:existence}. In a nutshell, the high-level idea of our proof relies on discretizing the cake --- such that every cell in the resulting grid has a very small value for each agent --- and computing the optimal outcome on the discretized cake using backward induction. At every cut step of the protocol, the grid is refined by adding a point between every two consecutive points of the grid from the previous cut step. This ensures that any ordering of the cut points that can be enforced by playing on the continuous cake can also be enforced on the discretized instance. Therefore, for the purpose of computing an approximate SPNE, it is sufficient to work with the discretization. We then show that the backward induction outcome from the discrete game gives an $\varepsilon$-SPNE on the continuous cake.

\subsection{Informed Tie-Breaking}
\label{subsec:ties}

Another approach for circumventing the example given at the beginning of the section is to change the \emph{tie-breaking} rule of Algorithm 1, by letting agent 1 choose even if $y=z$ (in which case agent 1 would cut in $x=0,y=1,z=1$, and get the entire cake). Tie-breaking matters: even the Dubins-Spanier protocol fails to guarantee SPNE existence due to a curious tie-breaking issue~\cite{BM13}.

To accommodate more powerful tie-breaking rules, we slightly augment GCC protocols, by extending their ability to compare cuts in case of a tie. Specifically, we can assume without loss of generality that the \emph{If-Else} statements of a GCC protocol are specified only with weak inequalities (as an equality can be specified with two inequalities and a strong inequality via an equality and weak inequality), which involve only pairs of cuts. We consider \emph{informed GCC protocols}, which are capable of using \emph{If-Else} statements of the form ``\emph{if} [$x<y$ or ($x=y$ and history of events $\in \mathcal{H}$)] \emph{then}''. That is, when cuts are made in the same location and cause a tie in an \emph{If-Else}, the protocol can invoke the power to check the entire history of events that have occurred so far. We can recover the $x<y$ and $x\leq y$ comparisons of ``uninformed'' GCC protocols by setting $\mathcal{H}$ to be empty or all possible histories, respectively. Importantly, the history can include where cuts were made exactly, and not simply where in relation to each other.


We say that an informed GCC protocol $\mathcal{P}'$ is \emph{equivalent up to tie-breaking} to a GCC protocol $\mathcal{P}$ if they are identical, except that some inequalities in the \emph{If-Else} statements of $\mathcal{P}$ are replaced with informed inequalities in the corresponding \emph{If-Else} statements of $\mathcal{P}'$. That is, the two protocols are possibly different only in cases where two cuts are made at the exact same point. 

For example, in Algorithm 1, the statement ``\emph{if} $x<y<z$ \emph{then}'' can be specified as ``\emph{if} $x<y$ \emph{then if} $y<z$ then''. We can obtain an informed GCC protocol that is equivalent up to tie-breaking by replacing this statement with ``\emph{if} $x<y$ \emph{then if} $y\leq z$ then'' (here we are not actually using augmented tie-breaking). In this case, the modified protocol may feel significantly different from the original --- but this is an artifact of the extreme simplicity of Algorithm 1. Common cake cutting protocols are more complex, and changing the tie-breaking rule preserves the essence of the protocol.

We are now ready to present our second main result.  

\begin{theorem}
	\label{thm:ties}
	For any $n$-agent GCC protocol $\mathcal{P}$ with a bounded number of steps and any $n$ valuation functions $V_1,\ldots,V_n$, there exists an informed GCC protocol $\mathcal{P}'$ that is equivalent to $\mathcal{P}$ up to tie-breaking, such that the game induced by $\mathcal{P}'$ and $V_1,\ldots,V_n$ has an SPNE.
\end{theorem}

Intuitively, we can view $\mathcal{P}'$ as being ``undecided'' whenever two cuts are made at the same point, that is, $x=y$: it can adopt either the $x<y$ branch or the $x>y$ branch --- there \emph{exists} an appropriate decision. The theorem tells us that for any given valuation functions, we can set these tie-breaking points in a way that guarantees the existence of an SPNE. In this sense, the tie-breaking of the protocol is \emph{informed} by the given valuation functions. Indeed, this interpretation is plausible as we are dealing with a game of perfect information. 

The proof of Theorem~\ref{thm:ties} is somewhat long, and has been relegated to \ref{app:ties}. This proof is completely different from the proof of Theorem~\ref{thm:existence}; in particular, it relies on real analysis instead of backward induction on a discretized space. The crux of the proof is the development of an auxiliary notion of \emph{mediated games} (not to be confused with Monderer and Tennenholtz's \emph{mediated equilibrium}~\cite{DM09}) that may be of independent interest. We show that mediated games always have an SPNE. The actions of the mediator in this SPNE are then reinterpreted as a tie-breaking rule under an informed GCC protocol. In the context of the proof it is worth noting that some papers prove the existence of SPNE in games with infinite action spaces (see, e.g., \cite{Harr85,HL87}), but our game does not satisfy the assumptions required therein.

%
%




\section{Fair Equilibria}
\label{sec:ef}

%

The existence of equilibria (Theorems~\ref{thm:existence} and \ref{thm:ties}) gives us a tool for predicting the strategic outcomes of cake cutting protocols. In particular, classic protocols provide fairness guarantees when agents act honestly; but do they provide any fairness guarantees in equilibrium?

We first make a simple yet crucial observation. In a proportional protocol, every agent is guaranteed a value of at least $1/n$ regardless of what the others are doing. Therefore, in every NE (if any) of the protocol, the agent still receives a piece worth at least $1/n$; otherwise it can deviate to the strategy that guarantees it a utility of $1/n$ and do better. Similarly, an $\varepsilon$-NE must be $\varepsilon$-proportional, i.e., each agent must receive a piece worth at least $1/n-\varepsilon$. Hence, classic protocols such as Dubins-Spanier, Even-Paz, and Selfridge-Conway guarantee (approximately) proportional outcomes in any (approximate) NE (and of course this observation carries over to the stronger notion of SPNE).

One may wonder, though, whether the analogous statement for envy-freeness holds; the answer is negative. We demonstrate this via the Selfridge-Conway protocol --- a 3-agent envy-free protocol, which is given in its truthful, non-GCC form as Algorithm 3. To see why the protocol is envy free, note that the division of three pieces in steps 4, 5, and 6 is trivially envy free. For the division of the trimmings in step 9, agent $i$ is not envious because it chooses first, and agent $j$ is not envious because it was the one that cut the pieces (presumably, equally according to its value). In contrast, agent 1 may prefer the piece of trimmings that agent $i$ received in step 9, but overall agent 1 cannot envy $i$, because at best $i$ was able to ``reconstruct'' one of the three original pieces that was trimmed at step 2, which agent 1 values as much as the untrimmed piece it received in step 6. 

\begin{algorithm}[t]
	\label{alg:SC}
	\begin{algorithmic}[1]
		\caption{Selfridge-Conway: an envy-free protocol for three agents.}
		\STATE Agent $1$ cuts the cake into three equal parts in the agent's value.
		\STATE Agent $2$ trims the most valuable of the three pieces such that there is a tie with the two most valuable pieces.
		\STATE Set aside the trimmings.
		\STATE Agent $3$ chooses one of the three pieces to keep.
		\STATE Agent $2$ chooses one of the remaining two pieces to keep --- with the stipulation that if the trimmed piece is not taken by agent $3$, agent $2$ must take it.
		\STATE Agent $1$ takes the remaining piece.
		\STATE Denote by $i \in \{2, 3\}$ the agent which received the trimmed piece, and $j = \{2, 3\}\setminus\{i\}$.
		\STATE Agent $j$ now cuts the trimmings into three equal parts in the agent's value.
		\STATE Agents $i$, $1$, and $j$ choose one of the three pieces to keep in that order.
	\end{algorithmic}
\end{algorithm}

We construct an example by specifying the valuation functions of the agents and their strategies, and arguing that the strategies are in SPNE. The example will have the property that the first two agents receive utilities of $1$ (i.e. the maximum value). Therefore, we can safely assume their play is in equilibrium; this will allow us to define the strategies only on a small part of the game tree. In contrast, agent 3 will deviate from its truthful strategy to gain utility, but in doing so will become envious of agent 1.

In more detail, suppose after agent 2 trims the three pieces we have the following.
\begin{itemize}
	\item Agent $1$ values the first untrimmed piece at $1$, and all other pieces and the trimmings at $0$.
	\item Agent $2$ values the second untrimmed piece at $1$, and all other pieces and the trimmings at $0$.
	\item Agent $3$ values the untrimmed pieces at $1/7$ and $0$, the trimmed piece at $1/14$, and the trimmings at $11/14$.
\end{itemize}
Now further suppose that if agent 3 is to cut the trimmings (i.e. take on the role of $j$ in the protocol), then the first two agents always take the pieces most valuable to agent $3$. Thus, if agent 3 does not take the trimmed piece it will achieve a utility of at most $1/7 + (11/14)(1/3) = 119/294$ by taking the first untrimmed piece, and then cutting the trimmings into three equal parts. On the other hand, if agent 3 takes the trimmed piece of worth $1/14$, agent 2 cuts the trimmings into three parts such that one of the pieces is worth $0$ to agent $3$, and the other two are equivalent in value (i.e. they have values $(11/14)(1/2) = 11/28$). Agents $1$ and $3$ take these two pieces. Thus, in this scenario, agent $3$ receives a utility of $1/14 + 11/28 = 13/28$ which is strictly better than the utility of $119/294$. Agent $3$ will therefore choose to take the trimmed piece. However, in this outcome agent $1$, from the point of view of agent $3$, receives a piece worth $1/7 + 11/28 = 15/28$ and so agent 3 will indeed be envious.

The foregoing example shows that envy-freeness is not guaranteed when agents strategize, and so it is difficult to produce envy-free allocations when agents play to maximize their utility. A natural question to ask, therefore, is whether there are any GCC protocols such that all SPNE are envy-free, and existence of SPNE is guaranteed. This remains an open question, but we do give an affirmative answer for the weaker solution concept of NE in the following theorem, whose proof appears in \ref{app:ef}.

\begin{theorem}
	\label{thm:ef}
	There exists a GCC protocol $\mathcal{P}$ such that on every cake cutting instance with strictly positive valuation functions $V_1,\ldots,V_n$, an allocation $X$ is the outcome of a NE of the game induced by $\mathcal{P}$ and $V_1,\ldots,V_n$ if and only if $X$ is an envy-free contiguous allocation that contains the entire cake.
\end{theorem}

Crucially, an envy-free contiguous allocation is guaranteed to exist~\cite{Strom80}, hence the set of NE of protocol $\mathcal{P}$ is nonempty. 

Theorem~\ref{thm:ef} is a positive result \emph{\`a la} implementation theory (see, e.g., \cite{Mask99}), which aims to construct games where the NE outcomes coincide with a given specification of acceptable outcomes for each constellation of agents' preferences (known as a \emph{social choice correspondence}). Our construction guarantees that the NE outcomes coincide with (contiguous) envy-free allocations, i.e. in this case the envy-freeness criterion specifies which outcomes are acceptable.

That said, the protocol $\mathcal{P}$ constructed in the proof of Theorem~\ref{thm:ef} is interesting theoretically, but it remains to be determined when its Nash equilibria arise in practice. This further motivates efforts to find an analogous result for SPNE. If such a result is indeed feasible, a broader, challenging open question would be to characterize GCC protocols that give rise to envy-free SPNE, or at least provide a sufficient condition (on the protocol) for the existence of such equilibria. 

%



\newpage

\appendix

\section{Proof of Theorem~\ref{thm:existence}}
\label{app:existence}

	Let $\varepsilon > 0$, and let $f(n)$ be an upper bound on the number of operations (i.e., on the height of the game tree) of the protocol. 
	Define a grid, $\mathcal{G}_1$, such that every cell on the grid is worth at most
	$\frac{\varepsilon}{2f(n)^2}$ to each agent.
	For every $n$, let
	$K$ denote the maximum number of cut operations, where $0 \leq K \leq f(n)$. For each $i \in \{1, \ldots, K\}$, we define the grid $\mathcal{G}_i$ so that the following properties are satisfied:
	\begin{itemize}
		\item The grids are nested, i.e.,
		$\{0, 1\} \subset \mathcal{G}_1 \subset \mathcal{G}_2 \subset \ldots \subset \mathcal{G}_K$.
		\item There exists a unique point $z \in \mathcal{G}_{i+1}$
		between any two consecutive points $x, y \in \mathcal{G}_i$, such
		that $x < z < y$ and $z \not \in \mathcal{G}_i$, for every $i\in\{1, \ldots, K-1\}$.
		\item
		Each cell on $\mathcal{G}_i$ is worth at most $\frac{\varepsilon}{2f(n)^2}$ to any agent, for all $i \in \{1, \ldots, K\}$.
	\end{itemize}
	
	Having defined the grids, we compute the backward induction outcome on the discretized cake, where the $i$-th \emph{Cut} operation can only be
	made on the grid $\mathcal{G}_i$. We will show that this outcome is an $\varepsilon$-SPNE, even though agents could deviate by cutting anywhere on the cake. 
	%
	%
	%
	On the continuous cake, the agents play a perturbed version of the
	idealized game from the grid $\mathcal{G}$, but
	maintain a mapping between the perturbed game
	and the idealized version throughout
	the execution of the protocol, such that each cut point from the continuous cake is mapped to a grid point that approximates it within a very small (additive) error. Thus when determining the next action, the agents use the idealized grid as a reference.
	The order of the cuts is the same in the ideal and perturbed game, however the values of the pieces may differ
	by at most $\varepsilon/f(n)$.
	
	
	We start with the following useful lemma. (For ease of exposition, in the following we refer to $[x,y]$ as the segment between points $x$ and $y$, regardless of whether $x < y$ or $y \leq x$.)
	
	\begin{lemma} \label{lem:map}
		Given a sequence of cut points $x_1, \ldots, x_k$ and nested grids $\mathcal{G}_1 \subset \ldots \subset \mathcal{G}_k$ with cells worth 
		at most $\frac{\epsilon}{4 f(n)^2}$ to each agent,
		there exists a map 
		$\mathcal{M}: \{x_1, \ldots, x_k\} \rightarrow \mathcal{G}_k$ such that:
		\begin{enumerate}[(1)]
			\item For each $i \in \{1, \ldots, k\}$, $\mathcal{M}(x_i) \in \mathcal{G}_i$.
			%
			\item The map $\mathcal{M}$ is order-preserving. Formally, for all $i,j \in \{1, \ldots, k\}$, $x_i < x_j$ $\iff$ $\mathcal{M}(x_i) < \mathcal{M}(x_j)$ and 
			$x_i = x_j$ $\iff$ $\mathcal{M}(x_i) = \mathcal{M}(x_j)$. 
			
			\item The piece $[x_i, \mathcal{M}(x_i)]$ is ``small'', that is: $V_l([x_i, \mathcal{M}(x_i)]) \leq \frac{k \epsilon}{ 2f(n)^2}$, for each agent $l \in N$.
			
		\end{enumerate}
	\end{lemma}
	\begin{proof}
		We prove the statement by induction on the number of cut points $k$.
		
		\emph{Base case}: We consider a few cases. If $x_1 \in \mathcal{G}_1$, then define $\mathcal{M}(x_1) := x_1$.
		Otherwise, let $R(x_1) \in \mathcal{G}_1$ be the leftmost point on the grid $\mathcal{G}_1$ to the right of $x_1$. 
		If $R(x_1) \neq 1$, define $\mathcal{M}(x_1) := R(x_1)$;
		else, let $L(x_1)$ denote the rightmost point on $\mathcal{G}_1$ strictly to the left of $1$ and
		define $\mathcal{M}(x_1) := L(x_1)$.
		To verify the properties of the lemma, note that:
		\begin{enumerate}[(1)]
			\item $\mathcal{M}(x_1) \in \mathcal{G}_1$.
			\item The map $\mathcal{M}$ is order-preserving since there is only one point.
			\item $V_l([x_1, \mathcal{M}(x_1)]) \leq \frac{\epsilon}{ 2f(n)^2}$ for each agent $l \in N$ since the grid $\mathcal{G}_1$ has (by construction)
			the property that each cell is worth at most $\frac{\epsilon}{2 f(n)^2}$ to each agent, and the interval $[x_1, \mathcal{M}(x_1)]$ is contained in a cell. 
		\end{enumerate}
		
		\emph{Induction hypothesis}: Assume that a map $\mathcal{M}$ with the required properties exists for any sequence of $k - 1$ cut points.
		
		\emph{Induction step}: Consider any sequence of $k$ cut points $x_1, \ldots, x_k$. By the induction hypothesis, we can map each cut point $x_i$ to a grid representative $\mathcal{M}(x_i) \in \mathcal{G}_i$, for all $i \in \{1, \ldots, k-1\}$,
		in a way that preserves properties 1--3. We claim that the map $\mathcal{M}$ on the points $x_1, \ldots, x_{k-1}$
		can be extended to the $k$-th point, $x_k$, such that the entire sequence $\mathcal{M}(x_1), \ldots, \mathcal{M}(x_k)$ satisfies the requirements of the lemma.
		We consider four exhaustive cases.
		\begin{enumerate}[(a)]
			\item There exists $i \in \{1, \ldots, k-1\}$ such that $x_k = x_i$. Then define $\mathcal{M}(x_k) := \mathcal{M}(x_i)$.
			
			\item There exists $i \in \{1, \ldots, k-1\}$ such that $x_i<x_k$, but $\cM(x_i)\geq x_k$. Let $x_j$ be the rightmost cut such that $x_j<x_k$; because $\cM$ is order-preserving, it holds that $\cM(x_j)\geq x_k$. Let $R(\cM(x_j))$ be the leftmost point on $\cG_k$ strictly to the right of $\cM(x_j)$, and set $\cM(x_k):=R(\cM(x_j))$. 
			
			Now let us check the conditions. Condition (1) holds by definition. Condition (2) holds because $\cM(x_k)>\cM(x_j)$, and for every $t$ such that $x_t>x_k$, $\cM(x_t)>\cM(x_j)$ and $\cM(x_t)\in \cG_{k-1}$, whereas $\cM(x_k)$ uses a ``new'' point of $\cG_k\setminus G_{k-1}$ that is closer to $\cM(x_j)$. For condition (3), we have that for every $l\in N$, 
			\begin{align*}
			&V_l([x_k,\cM(x_k)])\\
			&\leq V_l([x_j,\cM(x_k)])\\
			&= V_l([x_j,\cM(x_j)])+V_l([\cM(x_j),\cM(x_k)])\\
			&\leq\frac{(k-1)\epsilon}{2 f(n)^2} + \frac{\epsilon}{2f(n)^2}
			\leq  \frac{k \epsilon}{2 f(n)^2},
			\end{align*}
			where the third transition follows from the induction assumption.
			
			\item There exists $i \in \{1, \ldots, k-1\}$ such that $x_i>x_k$, but $\cM(x_i)\leq x_k$. This case is symmetric to case (b).
			
			\item For every $x_i$ such that $x_i<x_k$, $\cM(x_i)\leq x_k$, and for every $x_j$ such that $x_j>x_k$, $\cM(x_j)\geq x_k$. Let $x_i$ and $x_j$ be the rightmost and leftmost such cuts, respectively; without loss of generality they exist, otherwise our task is even easier. 
			
			Let $R(x_k)$ be the leftmost point in $\cG_k$ such that $R(x_k)\geq x_k$, and let $L(x_k)$ be the rightmost point in $\cG_k$ such that $L(x_k)\leq x_k$. Assume first that $\cM(x_j)>R(x_k)$; then set $\cM(x_k):=R(x_k)$. This choice obviously satisfies the three conditions, similarly to the base of the induction. 
			
			Otherwise, $R(x_k)=\cM(x_j)$ (notice that it cannot be the case that $R(x_k)>\cM(x_k)$); then set $\cM(x_k):=L(x_k)$. Let us check that this choice is order-preserving (as the other two conditions are trivially satisfied). Note that $\cM(x_j)\in \cG_{k-1}$, so $R(x_k)\in \cG_{k-1}$. Therefore, it must hold that $L(x_k)\in G_k\setminus G_{k-1}$ --- it is the new point that we have added between $R(x_k)$, and the rightmost point the left of it on $\cG_{k-1}$. Since it is also the case that $\cM(x_i)\in\cG_{k-1}$, we have that $\cM(x_i)<\cM(x_k)<\cM(x_j)$.  
		\end{enumerate}

		By induction, we can compute a mapping with the required properties for $k$ points. This completes the proof of the lemma.
	\end{proof}
	
	Now we can define the equilibrium strategies. Let $x_1, \ldots, x_{k}$ be the history of cuts made at some point during the execution
	of the protocol. By Lemma~\ref{lem:map}, there exists an order-preserving map $\mathcal{M}$ such that each point $x_i$ has a representative 
	point $\mathcal{M}(x_i) \in \mathcal{G}_i$ and the piece $[x_i, \mathcal{M}(x_i)]$ is ``small'', i.e. $$V_l([x_i, \mathcal{M}(x_i)]) \leq \frac{k\epsilon}{2 f(n)^2}\leq \frac{\epsilon}{2 f(n)}$$ for 
	each agent $l \in N$ --- using $k\leq f(n)$.

	Consider any history of cuts $(x_1, \ldots, x_k)$. Let $i$ be the agent that moves next.
	Agent $i$ computes the mapping
	$(\mathcal{M}(x_1), \ldots, \mathcal{M}(x_k))$.
	If the next operation is:
	\begin{itemize}
		
		\item \emph{Choose}:
		agent $i$ chooses the available piece (identified by the symbolic names of the cut points
		it contains and their order) which is optimal in the idealized game, given the current state and
		the existing set of ordered ideal cuts, $\mathcal{M}(x_1), \ldots, \mathcal{M}(x_k)$.
		Ties are broken according to a fixed deterministic scheme which is known to all the agents.
		
		\item \emph{Cut}:
		agent $i$ computes the optimal cut on $\mathcal{G}_{k+1}$, say at $x^{*}_{k+1}$. Then $i$ maps $x^{*}_{k+1}$ back to a point $x_{k+1}$
		on the continuous game, such that $\mathcal{M}(x_{k+1}) = x^{*}_{k+1}$. That is, the cut $x_{k+1}$ (made in step $k+1$)
		is always mapped by the other agents to $x^{*}_{k+1} \in \mathcal{G}_{k+1}$.
		Agent $i$ cuts at $x_{k+1}$.
	\end{itemize}
	
	We claim that these strategies give an $\varepsilon$-SPNE.
	The proof follows from the following lemma, which we show by induction on $t$ (the
	maximum number of remaining steps of the protocol):
	
	\begin{lemma} \label{lem:remainder_opt}
		Given a point in the execution of the protocol from which there are at most $t$ operations left 
		until termination,
		it is $\frac{t\varepsilon}{f(n)}$-optimal to play on the grid.
	\end{lemma}
	\begin{proof}
		Consider any history of play, where the cuts were made at $x_1, \ldots, x_k$. Without loss of generality, assume it is
		agent $i$'s turn to move.
		
		\emph{Base case:} $t=1$. The protocol has at most one remaining step. If it is a cut operation, then no agent receives
		any utility in the remainder of the game regardless of where the cut is made. Thus cutting on the grid ($\mathcal{G}_k$) is optimal.
		If it is a choose operation, then let $Z = \{Z_1, \ldots, Z_s\}$ be the set of pieces that $i$ can choose
		from. Agent $i$'s strategy is to map each piece $Z_j$ to its equivalent $\mathcal{M}(Z_j)$ on the grid $\mathcal{G}_k$, and choose
		the piece that is optimal on $\mathcal{G}_k$. Recall that $V_q([x_j, \mathcal{M}(x_j)])  \leq \frac{\epsilon}{2 f(n)}$ 
		for each agent $q \in N$.
		Thus if a piece is optimal on the grid, it is $\frac{\varepsilon}{f(n)}$-optimal in the continuous game (adding up the difference on both sides). It follows that $i$ cannot gain
		more than $\frac{\varepsilon}{f(n)}$ in the last step by deviating from the optimal piece on $\mathcal{G}_k$.
		
		\emph{Induction hypothesis:} Assume that playing on the grid is $\frac{(t-1)\varepsilon}{f(n)}$-optimal whenever there
		are at most $t-1$ operations left on every possible execution path of the protocol,
		and there exists one path that has exactly $t-1$ steps.
		
		\emph{Induction step:}
		If the current operation is \emph{Choose}, then by the induction hypothesis, playing on the grid in the remainder of the protocol
		is
		$\frac{(t-1)\varepsilon}{f(n)}$-optimal for all
		the agents, regardless of $i$'s move in the current step. Moreover,
		agent $i$ cannot gain by more than $\frac{\varepsilon}{f(n)}$ by choosing a different piece in the current step,
		compared to
		piece which is optimal on $\mathcal{G}_k$, since $V_{i}([x_{l}, \mathcal{M}(x_{l})]) \leq \frac{\varepsilon}{2f(n)}$ for all $l \in \{1, \ldots, k\}$.
		
		If the current operation is \emph{Cut}, then the following hold:
		\begin{enumerate}
			\item By construction of the grid $\mathcal{G}_{k+1}$, agent $i$
			can induce any given branch of the protocol using a cut in the continuous game if and only if the same branch
			can be induced using a cut on the grid $\mathcal{G}_{k+1}$.
			\item Given that the other agents will play on the grid for the remainder of the protocol, agent $i$ can change the size of at most one piece that it receives down the road
			by at most
			$\frac{\varepsilon}{f(n)}$ by deviating (compared to the grid outcome), since
			$V_{j}([x_{l}, \mathcal{M}(x_{l})]) \leq \frac{\varepsilon}{2f(n)}$ for all $l \in \{1, \ldots, k+1\}$ and for all $j \in N$.
		\end{enumerate}
		Thus by deviating in the current step, agent $i$ cannot gain more than $\frac{t\varepsilon}{f(n)}$. 
	\end{proof}
	Since $t\leq f(n)$, 
	the overall loss of any agent is bounded by $\varepsilon$ by Lemma~\ref{lem:remainder_opt}. 
	We conclude that playing on the grid is $\varepsilon$-optimal for all the agents, which completes the proof of the theorem. \qed

\section{Proof of Theorem~\ref{thm:ties}}
\label{app:ties}

Before we begin, we take this moment to formally introduce the auxiliary concept of a \emph{mediated game} in an abstract sense. We will largely distance ourselves from the specificity of GCC games here and work in a more general model. We do this for two purposes. First, it allows for a cleaner view of the techniques; and second, we believe such general games may be of independent interest. We begin with a few definitions.

\begin{definition}
	In an extensive-form game, an \emph{action tuple} is a tuple of actions that describe an outcome of the game. For example, the action tuple $(a_1, ..., a_r)$ states that $a_1$ was the first action to be played, $a_2$ the second, and $a_r$ the last.
\end{definition}

\begin{definition}
	\label{def:kspne}
	Given an action tuple, the $k^{th}$ action is said to be SPNE if the subtree of the game tree rooted where the first $k - 1$ actions are played in accordance to the action tuple is induced by some SPNE strategy profile. Furthermore, call such an action tuple \emph{$k$-SPNE}.
\end{definition}

Note that if the $k^{th}$ action is SPNE, so too are all actions succeeding it in the action tuple. To clarify Definition~\ref{def:kspne}, note that strategies of an extensive-form game are defined on every possible node of the game tree, so a $k$-SPNE action tuple can be equivalently defined as being an SPNE of the subgame rooted at the $k^{th}$ action.

With these definitions in hand, we can now describe the games of interest.

\begin{definition}
	We call an extensive-form game a \emph{mediated game} if the following conditions hold:
	\begin{enumerate}
		\item The set of agents consists of a single special agent, referred to as the \emph{mediator}, and some finite number $n$ of other regular agents. Intuitively, the mediator is an agent who is overseeing the proper execution of a protocol.
		\item The height $h$ of the game tree is bounded.
		\item Every agent's utility is bounded.
		\item Starting from the first or second action, the mediator plays every second action (and only these actions).
		\item Every action played by the mediator shares the same action space:
		$$\{0, ..., n\} \times \left([0, 1]^2 \cup 2^{\{1, ..., h\}}\right).$$
		This represents the agent who plays next ($0$ represents ending the game), and the interval which represents their action space  or the allowed pieces they may choose from.
		\item The mediator's utility is binary (i.e. it is in $\{0, 1\}$) and is described entirely by the notion of \emph{allowed edges}. This is a set of edges in the game tree such that the mediator's utility is $1$ iff it plays edges only in this set. Importantly, this set has the property that for every allowed edge, each grandchild subtree (i.e. subtree that represents the next mediator's action) must have at least one allowed edge from its root. Intuitively, these edges are the ones that follow the protocol the mediator is implementing.
		\item A regular agent's utility is continuous\footnote{The notions of convergence, compactness and continuity, which we will utilize often, necessarily assumes our action spaces are defined as metric spaces. Applicable metrics for the action spaces are not difficult to find, but are cumbersome to describe fully. We therefore will not belabour this point much further.} in the action tuple.
		\item \emph{Allowed-edges-closedness}: given a convergent sequence of action tuples where the mediator plays only allowed edges, the mediator must play only allowed edges in the limit action tuple as well.
	\end{enumerate}
\end{definition}

Note that appending meaningless actions (that affect no agent's utility) to a branch of the game tree will not affect the game in any impactful way. Thus, for the sake of convenience, we will assume for any game we consider all leaves of the game occur at the same depth (often denoted by $r$).

We now give a series of definitions and lemmas that culminate in the main tool used in the proof of Theorem~\ref{thm:ties}: all mediated games have an SPNE.


\begin{definition}
	A sequence of action tuples $\left(a^i_1, ..., a^i_r\right)\mid_i$ is said to be \emph{consistent} if for every $j$ the agent who plays action $a_j^i$ is constant throughout the sequence and, moreover, its action spaces are always subsets of $[0, 1]$ or always the same subset of $\{1, ..., h\}$ throughout the sequence.
\end{definition}

\begin{lemma}\label{thm:convergence}
	Let $\left(a^i_1, ..., a^i_r\right)\mid_i$ be a sequence of action tuples in a mediated game. Then there is a convergent subsequence.
\end{lemma}
\begin{proof}
	Due to the finite number of agents and bounded height of the game, we can find an infinite consistent subsequence $\vec{b}^i\mid_i = \left(b^i_1, ..., b^i_r\right)\mid_i$. It suffices to show this subsequence has a convergent subsequence of its own.
	It is fairly clear that we can find a convergent subsequence via compactness arguments, but there is a slight caveat: we must show that the limit action tuple is legal. That is, if the limit action tuple is $(a_1, ..., a_r)$ we must show that for every $i < r$ such that the mediator plays action $i$, action $i + 1$ is played by the agent prescribed by $a_i$, and within the bounds prescribed by it.
	We will prove this by induction.
	
	\emph{Base hypothesis:} First $0$ actions have a convergent subsequence --- this is vacuously true. 
	
	\emph{Induction hypothesis:} 
	Assume there exists a subsequence such that the first $k$ actions converge legally.
	
	\emph{Induction step}: 
	We wish to show that there exists a subsequence such that the first $k + 1$ actions converge. By the inductive assumption, there exists a subsequence $\vec{c}^i\mid_i$ such that the first $k$ actions converge. Now suppose $p$ plays the $k + 1^{th}$ action. If $p$ is the mediator, then the action space is indifferent to actions played previously and is compact. Thus, the $\vec{c}^i\mid_i$ must have a convergent subsequence such that the $k + 1^{th}$ element of the action tuple converges and so we are done. 
	
	Alternatively, if $p$ is a regular agent, the action space is not necessarily indifferent to previous actions. If the action spaces are always the same subset of $\{1, ..., h\}$, then we are clearly done. We therefore need only consider the case where the action spaces will be contained in $[0, 1]$. Due to the compactness of this interval, there will be a convergent subsequence of $\vec{c}^i\mid_i$ such that the $k + 1^{th}$ action converges to some $\gamma \in [0, 1]$. Call this subsequence $\vec{d}^i\mid_i$. 
	
	We argue that $\gamma$ is in the limit action space of the $k + 1^{th}$ action. For purposes of contradiction, assume this is false. Let $\delta$ be the length from $\gamma$ to the closest point in the limit action space (i.e. the action space in the limit given by the $k^{th}$ action played by the mediator). Then there exists some $M$ such that after the $M^{th}$ element in $\vec{d}^i\mid_i$, the closest point in the $k + 1^{th}$ action space to $\gamma$ is at least $\delta/2$ away. Moreover, there exists some $N$ such that after the $N^{th}$ element in $\vec{d}^i\mid_i$ the $k + 1^{th}$ action is no further than $\delta/3$ to $\gamma$. Elements of $\vec{d}^i\mid_i$ after element $\max(M, N)$ then simultaneously must have the $k + 1^{th}$ action space be at least $\delta/2$ away from $\gamma$ and have a point at most $\delta/3$ away from $\gamma$. This is a clear contradiction.
\end{proof}

\begin{lemma}\label{thm:tails}
	For every $k$, if we have a convergent sequence of action tuples where the $k^{th}$ action from the end is SPNE, then the $k^{th}$ action from the end for the limit action tuple is also SPNE. That is, for every $k$, convergent sequences of $(r - k + 1)$-SPNE action tuples are $(r - k + 1)$-SPNE.
\end{lemma}
\begin{proof}
	We prove the result by induction on $k$.
	
	\emph{Base Case $(k = 0)$:}
	This is vacuously true.
	
	\emph{Induction hypothesis $(k = m)$:} Assume convergent sequences of $(r - m + 1)$-SPNE action tuples are $(r - m + 1)$-SPNE.
	
	\emph{Induction step $(k = m + 1)$:} 
	Let $\vec{a}^i\mid_i = (a^i_1, ..., a^i_r)\mid_i$ be a convergent sequence of $(r - m)$-SPNE action tuples with the limit action tuple $(a_1, ..., a_r)$. We wish to show that if all actions before the last $m + 1$ actions play their limit actions, then the remaining $m + 1$ actions are SPNE --- note that by Lemma \ref{thm:convergence} we know that the limit sequence is a valid action tuple.
	
	Let $p$ be the agent that commits the $m + 1^{th}$ action from the end. If $p$ is the mediator, then by the definition of mediated games the desired statement is true (specifically via the allowed-edges-closedness condition). Now suppose instead that $p$ is not the mediator, and simply a regular agent. We show if the $m + 1^{th}$ action from the end took on some other valid value $\alpha \neq a_{r - m}$, there exists SPNE strategies for the remaining $m$ actions such that $p$ achieves a utility no higher than had it stuck with the limit action of $a_{r - m}$.
	
	So suppose the $m + 1^{th}$ action from the end in the $i^{th}$ element of the sequence is $\alpha^i$ such that $\lim_{i \rightarrow \infty}\alpha^i = \alpha$. Since $\vec{a}^i\mid_i$ is a sequence of $(r - m)$-SPNE action tuples, we can construct the sequence:
	\begin{equation*}
	\vec{b}^i\mid_i = (a^i_1, ..., a^i_{r - m - 1}, \alpha^i, \tilde{a}^i_1, ..., \tilde{a}^i_m)\mid_i
	\end{equation*}
	where the $\tilde{a}^i_j$ are SPNE actions such that $p$ achieves at most the utility achieved by instead playing $a^i_{r - m}$. Via Lemma \ref{thm:convergence}, $\vec{b}^i\mid_i$ must have a convergent subsequence --- call $\vec{c}^i\mid_i$ and indexed by increasing function $\sigma$.  That is, $\vec{c}^{i} = \vec{b}^{\sigma(i)}$. $\vec{c}^i\mid_i$ is then a convergent sequence of $(r - m + 1)$-SPNE action tuples and thus, by the inductive assumption, its limit action tuple is also an $(r - m + 1)$-SPNE.
	
	Now consider the limit action tuple $(a_1, ..., a_r)$ (of $\vec{a}^i\mid_i$) and the limit action tuple of $\vec{c}^i\mid_i$ denoted by $(c_1, ..., c_r)$. Note that:
	\begin{enumerate}
		\item $\forall i < r - m$: $a_i = c_i$.
		\item By the continuity requirement of mediated games (where $V_p$ is the utility function of $p$): 
		\begin{align*}
		&V_p(a_1, ..., a_r)\\
		&= \lim_{i \rightarrow \infty} V_p(a^i_1, ..., a^i_r)\\
		&= \lim_{i \rightarrow \infty} V_p(a^{\sigma(i)}_1, ..., a^{\sigma(i)}_r)\\
		&\geq \lim_{i \rightarrow \infty} V_p(a^{\sigma(i)}_1, ... a^{\sigma(i)}_{r - m - 1}, \alpha^{\sigma(i)}, \tilde{a}^{\sigma(i)}_{r - m + 1}, \ldots, \tilde{a}^{\sigma(i)}_r)\\
		&= \lim_{i \rightarrow \infty} V_p(c^i_1, ..., c^i_r)\\
		&= V_p(c_1, ..., c_r).
		\end{align*}
	\end{enumerate}
	These two points imply that we can set SPNE strategies for the remaining $m$ actions such that the utility of $p$ playing $\alpha$ is less than or equal to if it plays $a_{r - m}$ for the $m + 1^{th}$ action from the end (when the actions preceding the $m + 1^{th}$ action from the end are those given in the limit action tuple $(a_1, ..., a_r)$). As the $\alpha$ was arbitrary, the $m + 1^{th}$ action from the end of $(a_1, ..., a_r)$ can be made an SPNE action, which completes the proof.
\end{proof}


\begin{lemma}
	\label{thm:med_spne}
	All mediated games have an SPNE.
\end{lemma}
\begin{proof}
	We prove the lemma via induction on the height of the game tree. Note that this is possible as mediated games (like extensive-form games) are recursive: the children of a node of a mediated game are mediated games.
	
	\emph{Base case} (at most $0$ actions):
	This is vacuously true. 
	
	\emph{Induction hypothesis} (at most $k$ actions): 
	Assume we have shown that any mediated game with a game tree of height at most $k$ has an SPNE.
	
	\emph{Induction step} (at most $k + 1$ actions): 
	Let $p$ be the agent that commits the first action. If $p$ is the mediator, any action that is an allowed edge will be SPNE; and if no such action exists, any action will be SPNE (as the mediator is doomed to a utility of $0$). Now suppose $p$ is not the mediator.
	
	Assume by the inductive assumption, once $p$ makes its move, all remaining (at most) $k$ actions are SPNE actions. By the definition of a mediated game, $p$'s utility is bounded. Then the least upper bound property of $\mathbb{R}$ implies that $p$'s utility as a function of the first action must have a supremum $S$. Via the axiom of choice, we construct a sequence of possible actions for the first action that approaches $S$ in $p$'s utility. That is, we have some sequence $x^i\mid_i$ such that if $p$ plays $x^i$ for the first action, it achieves some utility $f(x^i)$ --- where $\lim_{i \rightarrow \infty} f(x^i) = S$. Moreover, let $g(x^i)$ map the action $x^i$ to a tuple of the remaining actions --- which are SPNE. By Lemma \ref{thm:convergence} $(x^i, g(x^i))\mid_i$ must have a convergent subsequence $(y^i, g(y^i))\mid_i$ that converges to $(y, g(y))$ --- where $y$ is a legal first action and $g(y)$ are legal subsequent actions.
	
	Notice that $(y_i, g(y_i))\mid_i$ is a convergent sequence of $2$-SPNE action tuples and thus by Lemma \ref{thm:tails}, $(y, g(y))$ is a $2$-SPNE action tuple as well. Furthermore, note that by the continuity requirement of mediated games, $y$ must give $p$ a utility of $S$. Therefore, this must be an SPNE action and so we are done.
\end{proof}

With this machinery in hand, we are now ready to complete the proof of Theorem \ref{thm:ties}. Our main task is to make a formal connection between mediated games and (informed) GCC protocols. 

\begin{proof}[Proof of Theorem~\ref{thm:ties}]
	Suppose we have a $n$-agent GCC protocol $\mathcal{P}$ with a bounded number of steps and and set valuations of the agents $V_1, ..., V_n$. Then we wish to prove that there exists an informed GCC protocol $\mathcal{P}'$ that is equivalent to $\mathcal{P}$ up to tie-breaking such that the game induced by $\mathcal{P}'$ and $V_1, ..., V_n$ has an SPNE.
	
	Outfit $\mathcal{P}$ as a game $M$, such that all but the final condition of mediated games are satisfied --- that is, the mediator enforces the rules of $\mathcal{P}$ and achieves utility $1$ if it follows the rules of $\mathcal{P}$ and $0$ otherwise. More explicitly, the mediator plays every second action and upon examination of the history of events (i.e. the ordering of the cuts made thus far, and results of choose queries), decides the next agent to play and their action space based on the prescription of $\mathcal{P}$. To see how all but the last condition is satisfied, we go through them in order.
	\begin{enumerate}
		\item This is by definition.
		\item The height of the tree is twice the height of the GCC protocol.
		\item The mediator's utility is bounded by $1$ by definition, and all other agent's utilities are bounded by $1$ as that is their value of the entire cake.
		\item This is by definition.
		\item When the mediator wishes to ask a \emph{Cut} query to agent $i$ in the interval $[a, b]$, it plays the action $(i, (a, b))$, whereas when it wishes to ask a \emph{Choose} query to agent $i$ giving them the choice between the $x_1^{th}, ..., x_k^{th}$ pieces from the left, it plays the action $(i, \{x_1, ..., x_k\})$. This method of giving choose queries deviates slightly from the definition given in Section \ref{sec:gcc}, but the two representations are clearly equivalent.
		\item The allowed edges are ones that follow the rules of $\mathcal{P}$.
		\item 
		This property is only relevant when considering \emph{Cut} nodes. To establish it, first consider the action in a single \emph{Cut} node, and fix all the other actions. We claim that for every $\epsilon>0$ there exists $\delta=\delta(\epsilon)>0$ \emph{that is independent of the choice of actions in other nodes} such that moving the cut by at most $\delta$ changes the values by at most $\epsilon$. 
		Indeed, let us examine how pieces change as the cut point moves. As long as the cut point moves without passing any other cut point, one piece shrinks as another grows. As the cut point approaches another cut point, the induced piece --- say $k$'th from the left --- shrinks. When the cut point passes another cut point $x$, the $k$'th piece from the left grows larger, or it remains a singleton and another piece grows if there are multiple cut points at $x$. In any case, it is easy to verify that the sizes of various pieces received in \emph{Choose} nodes change by at most $\delta$ if the cut point is moved by $\delta$. Furthermore, note that the number of steps is bounded by $r$ and --- since the value density functions are continuous --- there is an upper bound $M$ on the value density functions such that if $y-x\leq \delta'$ then $V_i([x,y])\leq M\delta'$ for all $i\in N$. Therefore, choosing $\delta\leq \epsilon/(Mr)$ is sufficient. Finally, $V_1,\ldots,V_n$ are continuous even in the actions taken in multiple \emph{Cut} nodes, because we could move the cut points sequentially. 
	\end{enumerate}
	
	We now alter $M$ such that at every branch induced by a comparison of cuts via an \emph{If-Else}, we allow in the case of a tie to follow either branch. Formally, suppose at a branch induced by the statement ``\emph{if} $x\leq y$ then $A$ \emph{else} $B$'' we now set in the case of $x = y$ the edges for both $A$ \emph{and} $B$ as allowed. Then we claim the property of allowed-edges-closedness is satisfied.
	
	To see this, let us consider action tuples. An action tuple where the mediator in $M$ only plays on allowed edges can be viewed as a trace of an execution of $\mathcal{P}$ which records the branch taken on every \emph{If-Else} statement --- though when there is a tie the trace may follow the ``incorrect'' branch. A convergent sequence of such action tuples at some point in the sequence must then keep the branches it chooses in the execution of $\mathcal{P}$ constant --- unless in the limit, the cuts compared in a branch that is not constant coincide. Thus, we have that in the limit, if a branch is constant, the mediator always takes an allowed edge trivially, and otherwise due to our modification of $M$ the mediator still takes an allowed edge. Furthermore, for all actions of the mediator that are not induced by \emph{If-Else} statements, the mediator clearly still plays on allowed edges and so we have proved the claim.
	
	Now as $M$ is a mediated game, it has an SPNE $S$ by Lemma \ref{thm:med_spne}. Let $\mathcal{P}'$ be the informed GCC protocol equivalent to $\mathcal{P}$ up to tie-breaking such that for every point in the game tree of $M$ that represents the mediator branching on an ``\emph{if} $x\leq y$ then $A$ \emph{else} $B$'' statement in the original protocol $\mathcal{P}$, $\mathcal{P}'$ chooses the $A$ or $B$ that $S$ takes in the event of a tie. Then the informedness of the tie-breaking is built into $\mathcal{P}'$ and we immediately see that the SPNE actions of the regular agents in $M$ correspond to SPNE actions in $\mathcal{P}'$.
\end{proof}

\section{Proof of Theorem~\ref{thm:ef}}
\label{app:ef} 

The proof of the theorem uses the Thieves Protocol given by Algorithm 4. In this protocol, agent $1$ first demarcates a contiguous allocation $X=\{X_1, ..., X_n\}$ of the entire cake, where $X_i$ is a contiguous piece that corresponds to agent $i$. This can be implemented as follows. First, agent $1$ makes $n$ cuts such that the $i$-th cut is interpreted as the left endpoint of $X_i$. The left endpoint of the leftmost piece is reset to $0$ by the protocol. Then, the rightmost endpoint of $X_i$ is naturally the leftmost cut point to its right or $1$ if no such point exists. Ties among overlapping cut points are resolved in favor of the agent with the smallest index; the corresponding cut point is assumed to be the leftmost one. Notice that every allocation that assigns nonempty contiguous pieces to all agents can be demarcated in this way.

After the execution of the demarcation step, $X$ is only a tentative allocation. 
Then, the protocol enters a verification round, where each agent $i$ is allowed 
to {\em steal} some non-empty strict subset of a piece (say, $X_j$) demarcated 
for another agent. If this happens (i.e., the if-condition is true) then agent $i$ 
takes the stolen piece and the remaining agents get nothing. This indicates the 
failure of the verification and the protocol terminates. Otherwise, the pieces 
of $X$ are eventually allocated to the agents, i.e., agent $i$ takes $X_i$.

We will require two important characteristics of the protocol. First, it guarantees that no state in which some agent steals can be a NE; this agent can always steal an even more valuable piece. Second, stealing is beneficial for an envious agent.

\begin{proof}[Proof of Theorem \ref{thm:ef}]
	Let $\mathcal{P}$ be the Thieves protocol given by Algorithm 3 and $\mathcal{E}$ be any NE of $\mathcal{P}$.
	Denote by $X$ the contiguous allocation of the entire cake obtained during the demarcation step, where $X_i = [x_i, y_i]$ for all $i \in N$, and let $w_i$ and $z_i$ be the cut points of agent $i$ during its verification round.
	Assume for the sake of contradiction that $X$ is not envy-free.
	Let $k^*$ be an envious agent, where $V_{k^*}(X_{j^*}) > V_{k^*}(X_{k^*})$, for some $j^* \in N$.
	There are two cases to consider: \\
	 
	\begin{algorithm}[t]
		\begin{algorithmic}
			\caption{Thieves Protocol: Every NE induces a contiguous envy-free allocation that contains the entire cake and vice versa.}
			\STATE Agent $1$ demarcates a contiguous allocation $X$ of the cake
			\FOR {$i = 2, \ldots, n, 1$}
			\STATE \textit{{// Verification of envy-freeness for agent $i$}}
			\STATE Agent $i$ \emph{Cuts} in $\{[0,1]\}$ // @$w_i$ 
			\STATE Agent $i$ \emph{Cuts} in $\{[w_i,1]\}$ // @ $z_i$
			\FOR{$j = 1$ to $n$}
			\IF{$\emptyset \not= ([w_i, z_i] \cap X_j) \subsetneq X_j$}
			\STATE \textit{{// Agent $i$ steals a non-empty strict subset of $X_j$}}
			\STATE Agent $i$ \emph{Chooses} from $\{[w_i, z_i] \cap X_j\}$
			\STATE \textbf{exit} \textit{// Verification failed: protocol terminates}
			\ENDIF
			\ENDFOR
			\STATE \textit{{// Verification successful for agent $i$}}
			\ENDFOR
			\FOR{$i = 1$ to $n$}
			\STATE Agent $i$ \emph{Chooses} from $\{X_i\}$
			\ENDFOR
		\end{algorithmic}
		\label{alg1}
	\end{algorithm}
	
	\emph{Case 1}: Each agent $i$ receives the piece $X_i$ in $\mathcal{E}$. This means that, during its verification round, each agent $i$ selects its cut points from the set $\bigcup_{j=1}^n\{x_j,y_j\}$. 
	By the non-envy-freeness condition for $X$ above (and by the fact that the valuation function $V_{k^*}$ is strictly positive), there exist $w'_{k^*}, z'_{k^*}$ such that $x_{j^*} < w'_{k^*} < z'_{k^*} < y_{j^*}$ and $V_{k^*}([w'_{k^*}, z'_{k^*}]) > V_{k^*}([x_{k^*}, y_{k^*}])$. Thus, agent $k^*$ could have been better off by cutting at points $w'_{k^*}$ and $z'_{k^*}$ in its verification round, contradicting the assumption that $\mathcal{E}$ is a NE. \\
	
	\emph{Case 2}: There exists an agent $i$ that did not receive the piece $X_i$. Then, it must be
	the case that some agent $k$ stole a non-empty strict subset $[w''_k,z''_k] = [w_k, z_k]\cap Z_j$ of another piece $X_j$. However,
	agent $k$ could have been better off at the node in the game tree reached in its verification round by making the following marks:
	$w_k' = \frac{x_j + w''_k}{2}$ and $z_k' = \frac{z''_k + y_j}{2}$. Since either $x_j \leq w_{k}'' < z_k'' < y_j$ or $x_j < w_k'' < z_k'' \leq y_j$ (recall that $[w''_k,z_k'']$ is a non-empty strict subset of $X_j$ and the valuation function $V_k$ is strictly positive), it is also true that $V_k([w_k', z_k']) > V_k([w_k'', z_k''])$, again contradicting the assumption that $\mathcal{E}$ is a NE. \\
	
	So, the allocation computed by agent $1$ under every NE $\mathcal{E}$ is indeed envy-free; this completes the proof of the first part of the theorem. \\
	
	We next show that every contiguous envy-free allocation of the entire cake is the outcome of a NE. Let $Z$ be such an allocation, with $Z_i=[x_i,y_i]$ for all $i\in N$. We define the following set of strategies $\mathcal{E}$ for the agents:
	\begin{itemize}
		\item At every node of the game tree (i.e., for every possible allocation that could be demarcated by agent 1), agent $i\geq 2$ cuts at points $w_i=x_i$ and $z_i=y_i$ during its verification round.
		\item Agent 1 specifically demarcates the allocation $Z$ and cuts at points $w_1=x_1$ and $z_1=y_1$ during its verification round.
	\end{itemize}
	Observe that $[w_i,z_i]\cap Z_j$ is either empty or equal to $Z_j$ for every pair of $i,j\in N$. Hence, the verification phase 
	is successful for every agent and agent $i$ receives the piece $Z_i$.
	
	We claim that this is a NE. Indeed, consider a deviation of agent 1 to a strategy that consists of the demarcated allocation $Z'$ (and the cut points $w'_1$ and $z'_1$). First, assume that the set of pieces in $Z'$ is different from the set of pieces in $Z$. Then, there is some agent $k \not= 1$ and some piece $Z'_j$ such that the if-condition $\emptyset \subset [x_k,y_k] \cap Z'_j \subset Z'_j$ is true. Hence, the verification round would fail for some agent $i\in \{2, ..., k\}$ and agent 1 would receive nothing. So, both $Z'$ and $Z$ contain the same pieces, and may differ only in the way these pieces are tentatively allocated to the agents. 
	But in this case the maximum utility agent 1 can get is $\max_{j}{V_1(Z'_j)}$, either by keeping the piece $Z'_1$ or by stealing a strict subset of some other piece $Z'_j$. Due to the envy-freeness of $Z$, we have: $$\max_{j}{V_1(Z'_j)} = \max_j{V_1(Z_j)} = V_1(Z_1),$$
	hence, the deviation is not profitable in this case either.
	
	Now, consider a deviation of agent $i\geq 2$ to a strategy that consists of the cut points $w'_i$ and $z'_i$. If both $w'_i$ and $z'_i$ belong to $\bigcup_{j=1}^n\{x_i,y_i\}$, then $[w'_i,z'_i] \cap Z_j$ is either empty or equal to $Z_j$ for some $j\in N$. Hence, the deviation will leave the allocation unaffected and the utility of agent $i$ will not increase. 
	If instead one of the cut points $w'_i$ and $z'_i$ does not belong to $\bigcup_{j=1}^n\{x_i,y_i\}$, this implies that the condition 
	$$\emptyset \subset [w'_i,z'_i]\cap Z_j \subsetneq Z_j$$ is true for some $j\in N$, i.e., agent $i$ will steal the piece $[w'_i,z'_i]\cap Z_j$. 
	However, the utility $V_i([w'_i,z'_i] \cap Z_j)$ of agent $i$ cannot be greater than $V_i(Z_j)$, which is at most $V_i(Z_i)$ due to the envy-freeness of $Z$. Hence, again, this deviation is not profitable for agent $i$.
	
	We conclude that $\cal E$ is a NE; this completes the proof of the theorem.
\end{proof}

\end{document}